\documentclass[10pt]{amsart} 
\usepackage{color} 

\newtheorem{Theorem}{Theorem}[section] 
\newtheorem{Lemma}[Theorem]{Lemma} 
\newtheorem{Cor}[Theorem]{Corollary}

\newtheorem{Prop}[Theorem]{Proposition}

\newcommand{\<}{\langle} 
\renewcommand{\>}{\rangle} 
\newcommand{\A}{A} 

\begin{document}  

\title[The {\v C}ern{\'y} conjecture]{The {\v C}ern{\'y} conjecture for automata \\ respecting intervals of a directed graph}  
\author{M. Grech and A. Kisielewicz}

\begin{abstract} The \v{C}ern\'{y}'s conjecture states that for every synchronizing automaton with $n$ states there exists a reset word of length not exceeding $(n-1)^2$.  We prove this conjecture for a class of automata preserving certain properties of intervals of a directed graph. Our result unifies and generalizes some earlier results obtained by other authors. 
\end{abstract} 

\maketitle

In this paper we consider finite (deterministic complete) automata $\A = \<Q,\Sigma,\delta \>$ with the state set $Q$, the 
input alphabet $\Sigma$, and the transition function $\delta : \; Q 
\times \Sigma \to Q$. The transition function defines the action of the letters in $\Sigma$ on $Q$, which, in this paper, is denoted simply by concatenation: $\delta(q,a) = qa$.  The action extends in a natural way to the words in $\Sigma^*$, and we use the same notation $qw = \delta(q,w)$. Accordingly, we write $Qw = \{qw : \; q\in Q\}$.

The automaton $\A$ is called \emph{synchronizing} if there exists a 
word $w\in\Sigma^*$ such that $|Qw|=1$ (in other words, $w$ \emph{resets} $\A$ sending all the states into one particular state). Such a word $w$ is called \emph{synchronizing} (or a \emph{reset} word)  for $\A$. The problem of synchronization is very natural and 
its various aspects are considered in the literature (see e.g. \cite{Epp, MS,Vol} for general information and further references). The most famous is the following conjecture due to \v{C}ern\'{y}.  
\bigskip 

\textbf{Conjecture} (Jan \v{C}ern\'{y} 1964, \cite{Cer}) \emph{If a deterministic finite $n$-state automaton $\A = \<Q,\Sigma,\delta \>$ is synchronizing, then it has a reset word of length $\leq (n-1)^2$. } 
\bigskip 

This conjecture is considered as one of the most longstanding open problems in the theory of finite automata. The consequent research includes verifying the conjecture for various classes of automata, establishing bounds for the length of reset words, investigating natural algorithmic and complexity questions, and many other related problems. For more detailed discussion we refer the reader to the most recent survey \cite{Vol} by Volkov. Here, we mention only the most important results proving the conjecture in special cases.  

In 1978, Pin \cite{Pin} proved the conjecture for \emph{circular} automata with a prime number of states (an automaton is \emph{circular} if it has a letter acting as a cyclic permutation of all the states). In 1990, Eppstein \cite{Epp} proved the conjecture for {orientable} automata (preserving a given cyclic ordering of all the states). In 1998, Dubuc \cite{Dub}, completing his earlier research, has verified the conjecture for all circular automata. In 2001, Kari \cite{Kar} verifies the conjecture for \emph{eulerian} automata (whose transition digraph is eulerian). In 2003, Ananichiev and Volkov \cite{AV} prove the conjecture for \emph{monotonic} automata (preserving a linear order of the states), and in 2005 they generalize their result to a broader class given by a certain multi-level construction \cite{AV1}. In 2007, Trahtman \cite{Tra1} demonstrates that the conjecture is true for \emph{aperiodic} automata (ones with the transition monoid having no nontrivial subgroups). In 2008, Almeida, Margolis, Steinberg, and Volkov \cite{AMSV} verify the conjecture for another class of automata related to the formal language theory: those with monoids belonging to \textbf{DS} class. In 2009, Volkov \cite{Vol1} proves the conjecture for the so-called \emph{weakly monotonic} automata, a certain strong generalization of generalized monotonic automata including the automata preserving a connected partial order. Most recently, in 2011, Steinberg \cite{Ste} verifies the conjecture for  automata having a letter inducing a connected digraph with the cycle of prime length (thus generalizing the mentioned Pin's result). In addition, various classes of small automata have been verified using computer programs. In  particular, Trahtman \cite{Tra}, in 2007, has announced checking all the automata on 2 letters with $n\leq 10$ states. 

Some of the above results involve assumptions on preserving by an automaton a certain structure. In this paper, we show that this kind of results may be unified and generalized. We introduce the notions of an \emph{interval} for a directed graph, \emph{respecting the intervals} of a digraph by an automaton, and the congruence induced by such a relation. Then, we prove a reduction theorem of the form that if an automaton $\A$ respects intervals of a directed graph and the induced quotient automaton satisfies the {\v C}ern{\'y} conjecture, then so does $\A$. Our result implies that the {\v C}ern{\'y} conjecture is true for a large class of automata that includes, in particular, orientable automata \cite{Epp}, monotonic and generalized monotonic automata \cite{AV,AV1}, aperiodic automata \cite{Tra1}, and weakly monotonic automata \cite{Vol1}.  

\section{Preliminaries}

There are two structures connected with an automaton $\A = \<Q,\Sigma,\delta \>$ giving  possibilities of viewing the automaton in various ways. The \emph{transition digraph} $T=T(\A)$ is defined on $Q$ by $(q,s)\in T$ if and only if $qa=s$ for some $a\in\Sigma$. The \emph{transition monoid} $M=M(\A)$ is one generated by transformations on $Q$ corresponding to letters in $\Sigma$. Then, the words in $\Sigma^*$ correspond to compositions of generating transformations. 

Automaton $\A = \<Q,\Sigma,\delta \>$ is \emph{strongly connected} if its transition digraph is strongly connected. In general, one can consider strongly connected components of $T(\A)$. These are partially ordered by the relation of the existence of a directed edge in $T(\A)$ from one component to another. It is clear that if $\A$ is synchronizing, then $T(\A)$ has to have a unique minimal strongly connected component $C$. In such a case, the restriction $\A|_{C}$ of $\A$ to $C$ forms a reduced automaton. It is easy to prove (and it is considered as a folklore result) that if $\A|_{C}$ satisfies the {\v C}ern{\'y} conjecture, then so does $\A$. Thus, it is enough to consider only strongly connected automata.  

In terms of monoids, this is just a reduction to transitive monoids, since it is clear that $\A$ is strongly connected if and only if $M(\A)$ is transitive. A further natural reduction  would be one to transition monoids which are \emph{primitive} in the sense of permutation group theory.  

Given  $\A = \<Q,\Sigma,\delta \>$, an equivalence relation $\sim$ on $Q$ is called a \emph{congruence} (on $\A$) if it is preserved by $\A$. More precisely, $q\sim s$ implies $qa\sim sa$ for all $a\in\Sigma$. If $\sim$ is a congruence, then the \emph{quotient automaton} $\A/\!\!\sim \; = \< Q', \Sigma, \delta'\>$ is defined in the natural way: $Q' = Q/\!\!\sim$ is the set of the equivalence classes of $\sim$, and $\delta([q]_{\sim},a)= [qa]_{\sim}$, where $[q]_{\sim}$ denotes the equivalence class containing $q$. One may conjecture that for all congruences the following equivalence holds: $\A$ satisfies the {\v C}ern{\'y} conjecture if and only if the quotient automaton does. If the {\v C}ern{\'y} conjecture is true then, of course, this equivalence is also true. Yet, we see no natural way to prove it, and according to our experience, this problem in general may be as hard as the {\v C}ern{\'y} conjecture itself. Yet, we have attempted to find special kinds of congruences, for which the problem would be more naturally linked to the {\v C}ern{\'y} conjecture and a reduction step would be possible. Here we have one easy and natural result.  

\begin{Prop}\label{quotient1} 
Let $\sim$ be a congruence on a strongly connected automaton  $\A = \<Q,\Sigma,\delta \>$ such that one of its equivalence class is a singleton. Then, $\A$ satisfies the {\v C}ern{\'y} conjecture, whenever $\A/\!\!\sim$ does.  
\end{Prop} 

\begin{proof} 
Let $|Q|=n$ and $k$ be the number of the equivalence classes in $\sim$. We may assume that $k<n$, since otherwise the result is trivial. By assumption, there exists a word $w\in\Sigma^*$ of length $\leq (k-1)^2$ resetting $\A/\!\!\sim$. This word sends all the states of $\A$ into one $\sim$-class $B$.  

Let $s$ be the unique element of a singleton $\sim$-class. Since $\A$ is strongly connected, there exists a word $u\in\Sigma^*$ that sends a fixed state $q\in B$ into $s$. Moreover, we may assume that the length $|u| \leq k-1$, since if  $u = a_1\ldots a_t$ is the shortest word with $qu=s$, then the path corresponding to applying successive letters of $w$ visits every $\sim$-class at most once. Indeed, suppose to the contrary that  $qa_1\ldots a_i$ and  $qa_1\ldots a_j$ are in the same $\sim$-class, and $i < j < t$. Then, since $\sim$ is a congruence, $qa_1\ldots a_{i}a_{j+1}j\ldots a_t = s$, a contradiction.

It follows that the word $u$ sends all the states of the $\sim$-class $B$ into $s$, and therefore $uw$ resets $\A$. The length $|uw| \leq (k-1)^2 + k-1  < (n-1)^2,$ as required.  
\end{proof} 

In the sequel, the argument used in this proof will be treated as routine.

\section{Intervals in digraphs} 

Let $D=\<Q,E\>$ be a directed graph. We say that a vertex $z\in D$ \emph{lies on a directed path} from $x$ to $y$ (in $D$) if there is a directed path $x=x_0, x_1, \ldots, x_n=y$ such that $z=x_i$ for some $i=0,1,\ldots,n$ and $x,y \notin \{x_1,\ldots,x_{n-1}\}$. In short, we say that the \emph{path $xy$ contains $z$}. Note that the path in question may have repeated occurrences of vertices except for endvertices $x,y$ each of which occurs only once. In the sequel, we assume that all considered paths have this property. Only if necessary, we stress this calling such paths \emph{singular}. 

The set of all vertices $z$ lying on a directed path between $x$ and $y$ will be denoted by $[x,y]$. Such a set is called a (directed) \emph{interval} of $D$. Note, that by definition, if $[x,y] \ne \emptyset$, then $x,y\in[x,y]$. Note also that due to the singularity assumption on paths, even for strongly connected digraphs (which are of special interest for us) the intervals $[x,y]$ and $[y,x]$ usually differs from each other. For example, if $D$ is a simple directed cycle, then the intervals correspond to a division of the cycle into two arcs. If $D$ is acyclic, then one of the intervals $[x,y]$ and $[y,x]$ is empty, and the other one is the usual interval in the induced partial ordering (the transitive closure of $D$). The interval $[x,x]$, by definition, is the strongly connected component of the vertex $x$ in $D$. In the sequel, we shall use the abbreviation \emph{sc-component}.

We say that an automaton $\A = \<Q,\Sigma,\delta\>$ \emph{respects the intervals} of a digraph $D$ defined on $Q$, if for every letter $a\in\Sigma$, and all $x,y\in Q$, the following conditions holds: 
\begin{itemize} 
\item[(i)] if $[x,y]\neq \emptyset$, then $[xa,ya]\neq \emptyset,$ 
\item[(ii)] if both $[x,y]\neq \emptyset$ and $[y,x]\neq \emptyset$, then $[x,y]a \subseteq [xa,ya],$ 
\item[(iii)] if $xa=ya$, then one of the sets $[x,y]a$ or $[y,x]a$ has at most one element.  
\end{itemize} 

The first condition means that $\A$ preserves existence of a directed path between vertices.   
The second condition means that in case when there are directed paths between $x$ and $y$ in both directions, $z\in[x,y]$ implies $za\in[xa,ya]$, which means that the relation of lying on a directed path between vertices is preserved. The third condition deals with the special case when the endpoints of the interval are mapped into the same vertex. What is happening with the two intervals $[x,y]$ and $[y,x]$ in such a situation? We assume a sort of continuity: when $x$ approaches $y$ then one of intervals is getting smaller while the other is getting larger.  

The first condition means, in particular, that $\A$ preserves both the partition into weakly connected components of $D$ and the partition into strongly connected components of $D$. Thus the corresponding equivalence relations are congruences, which will be referred to as congruences \emph{induced} by the weakly (or strongly) connected components of $D$, respectively.

\begin{Lemma}\label{cond} 
If $\A = \<Q,\Sigma,\delta\>$ respects the intervals of a digraph $D=\<Q,E\>$, then the conditions \mbox{\rm (i-iii)} above hold for every word $w$, as well $($that is, with $a$  replaced by $w$, and $\Sigma$ replaced with $\Sigma^*).$  
\end{Lemma} 

\begin{proof} 
For the condition (i) the claim is obvious. For (ii), we note that the condition implies that both $[xa,ya]$ and $[ya,xa]$ are nonempty. The claim for $w=a_1\ldots a_n$ follows by successive application of (ii) for letters $a_1,\ldots,a_n$. For (iii), if one of $[x,y]$ or $[y,x]$ is empty, then one of $[x,y]w$ or $[y,x]w$ is empty, as well, and thus satisfies (iii). Otherwise, successive application of (ii), yields that $[x,y]u \subseteq [xu,yu]$ for every prefix $u$ of $w$. Taking $u$ to be the largest prefix of $w$ such that $xu\neq yu$, and applying (iii) for $xu$ and $yu$, yields the required result for $w$.  
\end{proof} 

As an example, we leave to the reader to check that, if $D$ is a simple directed cycle, then respecting intervals means preserving an orientation of the cycle, while if $D$ is acyclic, then respecting intervals means preserving the partial order induced by $D$. We show that these are in fact the extreme cases of a general situation.  

Under assumptions as above, let $\sim$ denote the equivalence relation on $Q$ induced by sc-components of $D$, and let $Q' = Q/\!\!\sim$ be the quotient set. Then the relation induced on $Q'$ by $D$ is an acyclic digraph, which we denote by $D'= D/\!\!\sim$. The quotient automaton $\A/\!\!\sim$ respects the intervals of $D'$, that is, in view of the remark above, it preserves the partial order induced by $D'$. The property of respecting the intervals of a digraph may be now characterized as follows. 

\begin{Lemma} 
An automaton $\A = \< Q,\Sigma,\delta \>$ respects the intervals of a digraph $D$ on $Q$ if and only if  
$\A/\!\!\sim$ preserves the partial order induced by $D/\!\!\sim$ and for all $x,y$ belonging to the same sc-component of $D$ the conditions \mbox{\rm (ii)} and \mbox{\rm (iii)} above holds.  
\end{Lemma}

We say that the strongly connected components of a digraph $D$ are \emph{dense}, or that $D$ is \emph{scc-dense}, if for all $x,y,z$ belonging to the same sc-component $C$ of $D$ either $z\in[x,y]$ or $z\in[y,x]$. Note that a simple directed cycle is dense in that sense. The property is not as strong as it may seem at the first sight. Let us call a vertex $z\in C$ a \emph{check-point} for the pair $x,y\in C$ ($x,y\neq z$) if each directed path from $x$ to $y$ contains $z$ and each directed path from $y$ to $z$ contains $z$, as well. 

\begin{Lemma} 
A digraph $D$ is scc-dense if and only if contains no check-point in any sc-component of $D$.  
\end{Lemma} 

\begin{proof} 
Suppose first $D$ has a check-point $z$ in an sc-component $C$. So, there are $x,y\in C$ such that any directed path between $x$ and $y$ goes through $z$. It follows that $y\notin [z,x]$ and $y\notin [x,z]$.  

Conversely, suppose that $z\notin[x,y]$. Then, one of the following holds: either  
(a) every path $xz$ contains $y$, or (b) every path $zy$ contains $x$.  
 
Similarly, supposing that $z\notin[y,x]$, we have that either  
  (c) every path $yz$ contains $x$, or   (d) every path $zx$ contains $y$.  
    
   Now, (a) and (c) contradicts the fact that there is a path form $x$ to $z$, while (a) and (d) means that $y$ is a check-point. The situation for the remaining two possibilities is similar. Consequently, there is always a check-point, as required. 
\end{proof}

Now we prove a crucial property of the automata respecting the intervals of an scc-dense digraph connected with synchronization. We will be interested in the situation when the image $Xw$ of the subset $X$ of $Q$ under the word $w$ \emph{collapses}, by which we mean that $|Xw|=1$.  

\begin{Lemma}\label{crucial} 
Let $\A = \< Q,\Sigma,\delta \>$ respects the intervals of an scc-dense digraph $D$ on $Q$. 
If $X$ is a nonempty subset of $Q$, $|X|>1$, contained in an sc-component of $D$ such that $|Xw|=1$, then there exists different $x,y\in X$ such that $X\subseteq [x,y]$ and the whole $[x,y]$ collapses under $w$.  
\end{Lemma} 
\begin{proof} 
The proof is by induction on the cardinality of $X$. For $|X| = 2$ the claim follows trivially from the condition~(iii) of Lemma~\ref{cond}. Otherwise, let us fix $z\in X$. By induction assumption, there are different $x, y \in X \setminus \{z\}$ such that $X \setminus \{z\} \subseteq [x, y]$ and  $[x, y]$ collapses under $w$. If $z\in [x, y]$, then we are done. 
So, we may assume $z\notin[x,y]$. 

Since $D$ is scc-dense, it follows that $z\in[y,x]$.  
This means that there is a path $yx$ containing $z$, and consequently, there is a path $yz$ containing no $x$. This path may be added to every path $xy$ yielding a path from $x$ to $z$. It follows that $[x,y] \subseteq [x,z]$. Similarly (since there is a path $zx$ containing no $y$),$[x,y] \subseteq [z,y]$.  

If any of $[x,z]$ or $[z,y]$ collapses under $w$, we are done. Hence, we may assume, by condition (iii) of Lemma~\ref{cond}, that both $[z,x]$ and $[y,z]$ collapse under $w$, that is, 

$$|[z,x]w|=|[y,z]w| = 1. $$   

We may assume also that there are $t\in [x,z]$, such that $tw\neq zw=xw$, and  
$s\in [z,y]$, such that $sw\neq zw = yw$.  

We show that among paths $xz$ containing $t$ there exists one not containing $y$. Suppose to the contrary that each path $xz$ containing $t$ contains $y$, as well. If the segment $xt$ of such a path does not contain $y$, then neither does the segment $tz$ (otherwise, we would have a path $xy$ containing $t$, and since by induction hypothesis $|[x,y]w| = 1$, it follows that $tw=xw$, a contradiction). Similarly, if the segment $tz$ of such a path does not contain $y$, then neither does the segment $xt$ (otherwise, we would have a path $yz$ containing $t$, and since, by assumption above, $[y,z]$ collapses, we have $tw=zw$, a contradiction). It follows that every path $xt$ contains $y$, and every path $tz$ contains $y$. By singularity, there is a path $zt$ not containing $y$, and therefore not containing $x$ either. Similarly, there is a path $tx$ not containing $y$, and therefore not containing $z$ either. This yields a (singular) path $zx$ containing $t$, which (since $[z,x]$ collapses) yields $tw=tz$, a contradiction.  

In a completely analogous way we show that among paths $zy$ containing $s$ there exists one not containing $x$. Combining it with a path $xz$ containing $t$ and not containing $y$, we obtain a (singular) path $xy$ containing $z$, which contradicts the assumption that $z\notin[x,y]$, and thus completes the proof. 
\end{proof}

\section{Main result}

We start from a preliminary result. 
Given a synchronizing automaton $\A = \<Q,\Sigma,\delta \>$ and a subset $C$ of $Q$.  
A family $\mathcal F=\{X_1,\ldots, X_m\}$ of subsets of $Q$ is called a \emph{{\v C}ern{\'y} family} for $C$, if for every image $Y=Cu$ by a word $u\in\Sigma^*$ and every $w\in \Sigma^*$, if $|Y|>1$ and $|Yw|=1$ then there exists $i\leq m$ such that $Y\subseteq X_i$ and $|X_iw|=1$. 
In other words, if $C$ collapses under $uw$, then the image $Cu$, if nontrivial, may be extended to a member of $\mathcal F$ that still collapses under $w$. 

\begin{Lemma} \label{Cfamily}
If $\mathcal F=\{X_1,\ldots, X_m\}$ is a {\v C}ern{\'y} family for a set $C\subset Q$, $|C|>1$, then there exists $w\in \Sigma^*$ of length not exceeding  $m =|\mathcal F|$ such that $|Cw|=1$. 
\end{Lemma} 

\begin{proof} 
Since $\A$ is synchronizing, there is a word $v=a_1\cdots a_t$ such that $Cv$ collapses under $v$. We assume that $v$ is the shortest possible word with this property.  
Denote $Y_i=Qa_1 \cdots a_i$ for $i \in \{1, \ldots, t-1\}$, and $Y_0=C$. Then, for all $i$, $|Y_i|>1$.  
For each $i \in \{0, \ldots, t-1\}$ choose $Z_i \in \mathcal F$ such that $Y_i \subseteq Z_i$ and $|Z_i a_{j+1} \cdots a_t|=1$.  
We show that if $i<j$, then $Z_i \ne Z_j$.  
Indeed, if $i<j$ and $Z_i=Z_j$, then we have $Ca_1 \cdots a_i = Y_i \subseteq Z_j$, and consequently,  $|Ca_1 \cdots a_i a_{j+1} \cdots a_t|=1$, contradicting the fact that $v$ is the shortest word with this property. Thus,  $t \leq m$, as required. 
\end{proof}

Our main result is the following

\begin{Theorem}\label{main} 
Let $\A = \< Q,\Sigma,\delta \>$ be a strongly connected automaton respecting the intervals of  an scc-dense digraph $D$ on $Q$, and let $\sim$ denote the congruence on $\A$ induced by weakly connected components of $D$. Then, if $\A/\!\!\sim$ satisfies the {\v C}ern{\'y} conjecture, then so does $\A$.  
\end{Theorem}

\begin{proof} 
Denote by $n$ the number of strongly connected components, and by $k$ the number of  
weakly connected components (wc-components). If there is a wc-compo\-nent of $D$ consisting of a single vertex, then by Proposition~\ref{quotient1} the theorem is true. Hence we may assume that, each wc-component has at least two vertices.  
Assume also that $\A/\!\!\sim$ is synchronizing and satisfies the {\v C}ern{\'y} conjecture.  
Recall that, since $\A$ respects the intervals, wc-components are mapped by transformations of $\A$ into wc-components, and sc-components are mapped into sc-components.  
We first show that \bigskip 

\noindent \emph{Claim 1.  
There exists a word $v\in\Sigma^*$ of length not exceeding $(n-1)^2$ such that $Qv$ is contained in a single sc-component of $D$.}  
\bigskip 

By assumption on $\A/\!\!\sim$, there is a word $w$ of the length at most $(k-1)^2$ that resets $\A/\!\!\sim$. 
The word $w$ maps all the states of $\A$ into one wc-component $B$ of $D$.  
Since $\A/\!\!\sim$ is strongly connected there is a word $u$ of the length at most $k-1$  mapping  $B$ on any chosen wc-component $B_i$. In particular, if there exists a wc-component consisting of a single sc-component, we are done. Indeed, in such a case $wu$ maps $Q$ into a single sc-component and it has length $$(k-1)^2 + (k-1) = (k-1)k < (n-1)^2.$$ 
(We note that for $k=1$, $w$ is simply the empty word, and the argument works, as well).

Thus, we may assume that every wc-component of $D$ has at least two sc-components. 
The sc-components are partially ordered by $D$, so we may speak of maximal and minimal sc-components. 
Let $B_1, \ldots, B_k$ denote the wc-components of $D$ (this set includes $B$ defined above). Then, by $M_i$ we denote the number of maximal sc-components in $B_i$, and by $m_i$ the number of minimal sc-components of $B_i$. Since each $B_i$ has at least two sc-components, no sc-component is both maximal and minimal. Therefore  $\sum_{i=1}^k (M_i+m_i) \leq n.$ Let us define 
$$p=\min\{M_i,m_i: i \in \{1, \ldots, k\}\}.$$  
Then, $p \leq {n}/{2k}$. To fix attention, we may assume without loss of generality that $p = M_1$, and that the word $u$  
of length at most $k-1$, mentioned above, maps $B$ into $B_1$.

Let $C_1, \ldots, C_p$ denote the maximal sc-components of $B_1$.  
We define inductively a sequence of words $u_0,\ldots,u_t$, with $t\leq p$, and such that $B_1u_t$ is contained in a single sc-component of $D$. Let $u_0$ be the empty word. Given the word $u_i$, we define $J_i$ to be the set of those $j \in \{1, \ldots, p\}$ for which none of the images $C_ju_1, C_ju_2, \ldots, C_ju_i$ is contained in a minimal sc-component of $D$. Further, let $X_i $ be the set of those sc-components of $D$ that contain $C_ju_i$ for some $j \in J_i$. In particular, we have $J_0 = \{1,2,\ldots,p\}$ and $X_0 = \{C_1, \ldots, C_p\}$.  
The word $u_{i+1}$ is now defined as follows.  

We look first for a word $v_{i+1}$ that maps at least one sc-component in $X_i$ into some minimal sc-component of $D$, Such a word exists, since $A$ is strongly connected. We may assume that the length of  
$v_{i+1}$ is at most $n-(m-1)-|X_i|$, where $m=\sum_{j=1}^k m_i$. Indeed, let $I_1,I_2,\ldots,I_r$ be the sequence of images of an sc-component in $X_i$ obtained by applying successive letters of word $v_{i+1}$. Obviously, each such image is contained in an sc-component, so we have a sequence of corresponding sc-components. We may assume that in this sequence no sc-component appears twice, no sc-component of $X_i$ occurs, and except for the last term, no minimal sc-component occurs. This yields the length bound above.  
We define $u_{i+1} = u_iv_{i+1}$. Then obviously $|X_{i+1}| < |X_i|$, unless $X_i$  is empty. It follows that $X_t$ is empty for some $t \leq p$, and consequently, for each $j\in \{1,2,\ldots,p\}$ there exists $i \leq t$ such that $C_ju_i$ is contained in a minimal sc-component of $D$. 
\bigskip 

\noindent \textit{Claim 2. The image $B_1u_t$ is contained in a single sc-component.}\bigskip 

Let $\preceq$ denotes the partial ordering of sc-components in $D$. Since $B_1$ is weakly connected, it is enough to show that for every pair of sc-components $Z,Y \in B_1$ such that $Z$ covers $Y$ in $\preceq$, the images $Zu_t$ and  $Yu_t$ are contained in the same sc-component. To this end, let $C_j \in X_0$ be the maximal sc-component of $B_1$ with $Z \preceq C_j$.  
Now, since for some $i$, $C_ju_i$ is contained in a minimal sc-component, and $\A$ preserves $\preceq$, the images $Zu_i$ and $Yu_i$ are contained in the same minimal sc-component. Since $u_t = w_1u_iw_2$, for some $w_1,w_2$, the claim follows. 
\bigskip

Thus, the word $v=wuu_t$ maps $Q$ into a single sc-component. Let us denote this component by $C$. To complete the proof of {Claim~1} we need to prove now the following. \bigskip

\noindent \textit{Claim 3. The length $|v| \leq (n-1)^2$. }\bigskip

By construction, we have $$ |v| \leq (k-1)^2 + k-1 + \sum_{i=0}^{p-1} (n-m -p+i +1)\leq k(k-1) + pn - kp^2 - \frac{p(p-1)}{2}.$$ (In the later inequality, we have used the fact that $m=\sum_{j=1}^k m_i \geq kp$). 

If $p=1$, then  
$$|v| \leq k(k-2) + n  \leq \frac{n}{2}\left(\frac{n}{2}-2\right) + n \leq (n-1)^2.$$ 
So, we may assume that $p \geq 2$.  
Since $2kp\leq n$, we have  
$$p(n-kp) \leq \frac{n^2}{2k} - \frac{np}{2} \leq \frac{n^2}{2k} - n.$$  
It follows that  
$$|v| \leq k(k-1) + \frac{n^2}{2k} - n - 1.$$ 

Consider the function $f(k)=k(k-1) + ({n^2}/{2k}) - n - 1$.  
We know that $1 \leq k \leq {n}/{2}$, yet we consider it on the interval $[1,n-1]$.  
Since $f'(k)$ is continuous and increasing on the interval $[1,n-1]$, $f(k)$ have the maximal value, in this interval, either for $k=1$ or for $k=n-1$.  
For $k=1$, we have  
$$|v| \leq \frac{n^2}{2} - n - 1 < (n-1)^2.$$ 
For $k=n-1$, we have  
$$|v| \leq (n-1)(n-2)+ \frac{n^2}{2(n-1)} - n - 1.$$ 
Since, $n\geq 2kp \geq 4$, we have $n \leq 2(n-1)$, and consequently 
$$|v| \leq (n-1)(n-2) - 1 < (n-1)^2.$$ 
This completes the proofs of Claim~3 and Claim~1. 
\smallskip

Now, if all the sc-components in $D$ are trivial (one-element), then the proof of the theorem is finished.  
Otherwise, let $N$ denote the number of vertices in $D$ (states in $A$). Then $N>n$. 
If there is at least one trivial sc-component consisting of a single state $q$, then there is a word $v_0$ of length $\leq (n-1)$ mapping the sc-component $C$ containing $B_1u_t$ onto $q$. Consequently, the word $vv_0$ resets $A$ and it has the length not exceeding 
$$(n-1)^2 + (n-1) = n(n-1) < (N-1)^2,$$ 
as required. Hence, we may assume that each sc-component contains at least two vertices. Then $N\geq 2n$.

We show that there exists a small enough {\v C}ern{\'y} family of sets for $C$. Suppose that $Y=Cu$ is an image of $C$ by a word $u$. Then, $Y$ is contained in some sc-component of $D$. 
Let $w$ be such that $|Yw|=1$. By Lemma~\ref{crucial}, there exist $x,y\in Y$ such that $Y \subseteq [x,y]$ and $[x,y]$ collapses under $w$. This means that the family $\mathcal F$ of sets $[x,y]$ with $x\neq y$ belonging to the same sc-component in $D$ forms a {\v C}ern{\'y} family of sets for $C$. We estimate the cardinality $m =|\mathcal F|$ of this family.  

Let $c_i = |C_i|$. Then $m =  \sum_{i=1}^n c_i(c_i -1)$. Clearly, for a fixed $n$ and $N$, this sum has the maximal value when $c_1$ is as large as possible and the remaining values of $c_i$ are as small as possible. Since, by assumption $c_i > 1$, we have 
$$m =  \sum_{i=1}^n c_i(c_i -1) \leq (N-2(n-1))(N-2n+1) + 2(n-1),$$ 

and consequently, 

$$m  \leq N^2-4Nn+4n^2+3N-4n.$$
   
This may be written as
 
$$m \leq (N-1)^2 - (n-1)^2 -(4Nn-5n^2-5N+6n).$$ 

The latter term in parentheses equals 
$$ (3N-5n)(n-1)+Nn - 2N +n, $$

and since $N\geq 2n$ and $n>1$, it is positive. Therefore, 

$$m \leq (N-1)^2 - (n-1)^2.$$ 

Combining Lemma~\ref{Cfamily} with Claim~1 completes the proof. 
\end{proof} 

\section{Corollaries and applications}

First note that if an scc-dense graph $D$ on $Q$ is weakly connected, then  $\A/\!\!\sim$ in Theorem~\ref{main} has one element and trivially satisfies the {\v C}ern{\'y} conjecture (see a remark at the beginning of the proof of Theorem~\ref{main} concerning the case $k=1$). Therefore we have the following 

\begin{Cor}
If a strongly connected automaton $\A = \< Q,\Sigma,\delta \>$ respects the intervals of an scc-dense weakly connected digraph $D$ on $Q$, then $\A$ satisfies the {\v C}ern{\'y} conjecture.  
\end{Cor} 

Obviously, a (simple directed) cycle is scc-dense, and so, as a particular case of the corollary above, we obtain Eppstein's result \cite{Epp} that oriented automata (called in \cite{Epp} \emph{monotonic}) satisfy the  {\v C}ern\'{y} conjecture. The following is a natural generalization of the Eppstein result.

Let us call a directed graph $D$ a digraph with \emph{unique return paths}, if whenever there is a path from $x$ to $y$ in $D$, and there is a return path from $y$ to $x$, then the latter is unique, and in consequence the former is also unique. Since we allow on paths repeated occurrences of vertices (other than endvertices), it means that on each of the paths every next step is uniquely determined until we reach the endvertex. This leads easily to the conclusion that every two cycles in $D$ are disjoint. Thus, a digraph with unique return paths consists of a collection of disjoint cycles and, possibly, additional edges between them inducing a partial ordering on the cycles. In a sense, this is a class of digraphs next to the acyclic digraphs (and thus a natural generalization of simple cycles). Note that preserving intervals of $D$ in this case means that cycles are mapped into cycles, orientation of vertices on each cycle is preserved and the existence of a directed path between cycles is preserved. We abbreviate it saying that the induced partial order and orientation of cycles are preserved. Then we may formulate the following simple generalization of Eppstein result.

\begin{Cor}
Let $D$ be a weakly connected digraph on $Q$ with unique return paths. If a strongly connected automaton $\A = \< Q,\Sigma,\delta \>$ preserves the partial ordering on the cycles induced by $D$ and orientation of the cycles, then $\A$ satisfies the {\v C}ern{\'y} conjecture.  
\end{Cor}

Now, it is easy to construct examples of automata (using, for instance, {\v C}ern\'{y} automata as components) that are neither orientable nor weakly monotonic and satisfy the {\v C}ern\'{y} conjecture due to the corollary above.

We proceed to show that our results also include all weakly monotonic automata. This leads to a different (and perhaps simpler) way to handle weakly monotonic automata and to a generalization of these constructions. First let us recall the definition from \cite{Vol1}. Given an automaton $\A = \< Q,\Sigma,\delta \>$, a binary relation $\rho$ on  $Q$ is \emph{stable} if $(p,q)\in\rho$ implies $(pa,qa)\in\rho$ for all $p,q\in Q$ and $a\in\Sigma$. Then,  $\A$ is called \emph{weakly monotonic of level} $\ell\geq 1$ if there exists a strictly increasing chain of stable binary relations 
$$\rho_0 \subset \rho_1 \subset \ldots \subset \rho_{\ell}, $$
such that $\rho_0$ is the equality relation, the transitive closure of $\rho_{\ell}$ is the universal relation, and for each $i=0,\ldots,\ell-1$ the following condition is satisfied:

\begin{itemize} 
\item[($*$)]  the transitive closure $\pi_{i}$ of $\rho_{i}$ is contained in $\rho_{i+1}$ and the relation $\rho_{i+1}/\pi_{i}$ induced by $\rho_{i+1}$ on the set of equivalence classes $Q/\pi_{i}$ is a partial ordering. 
\end{itemize}

Note that the equivalence classes of the transitive closure of a partial ordering relation  $\tau$ are just the weakly connected components of any digraph $D$ inducing $\tau$. Recall that if $\A= \< Q,\Sigma,\delta \>$ preserves the partial order $\tau$ (defined on $Q$) then the weakly connected components of $\tau$ form a congruence on $\A$. It is natural to write simply $\A/\tau$ to denote the corresponding quotient automaton. Then, we may define weakly monotonic automata by recursion as follows. 

\begin{itemize} 
\item[{\rm (i)}] The trivial automata with one state are weakly monotonic of level $0$; 
\item[{\rm (ii)}] for $\ell>0$, $\A = \< Q,\Sigma,\delta \>$ is a weakly monotonic of level $\ell$, if 
there exists a partial order $\tau$ on $Q$, such that $\A/\tau$ is a weakly monotonic automaton of level $\ell-1$.
\end{itemize} 

$\A = \< Q,\Sigma,\delta \>$ is a weakly monotonic automata of level $1$, if 
if there exists a \emph{weakly connected} partial order $\tau$ on $Q$ preserved by $\A$; 

It should be clear that in both the definitions the classes of weakly monotonic automata of level not exceeding $\ell$ coincide. In particular, the weakly monotonic automata of level $1$ are just the automata preserving a \emph{connected} partial order.
The fact that each strongly connected weakly monotonic automaton of level $\ell$ satisfies the {\v C}ern{\'y} conjecture follows now by $\ell$-fold application of Theorem~\ref{main}. We note, that Volkov \cite{Vol1} proved this result for all (not necessarily strongly connected) weakly monotonic automata. This can be also obtained using the proof of Theorem~\ref{main} (since the assumption that the automaton in question is strongly connected is necessary only for the case when $D$ has cycles). It should be noted however, that in \cite{Vol1}, for strongly connected automata a stronger bound for a reset word is established. It has been also observed in \cite{Vol1} that Trahtman \cite{Tra1} proved that each aperiodic automaton preserves a nontrivial partial order, and therefore each such automaton is weakly monotonic (the latter can be seen immediately from our recursive definition).  

These observations may be generalized for strongly connected automata as follows. For an automaton 
$\A= \< Q,\Sigma,\delta \>$ respecting the intervals of a digraph $D$, by $\A/D$ we denote the quotient automaton of the congruence induced by the weakly connected components of $D$. For an arbitrary class of strongly connected automata $\mathcal C$, we define recursively the class $I_{\ell}({\mathcal C})$: 

\begin{itemize} 
\item[{\rm (i)}] $I_0(\mathcal C) = {\mathcal C}$; 
\item[{\rm (ii)}] for each $\ell > 1$,  $I_{\ell}{(\mathcal C)}$ is the class of all strongly connected automata $\A=\< Q,\Sigma,\delta \> $ such that, for some digraph $D$ on $Q$, $\A$ respects the intervals of $D$ and $\A/D \in I_{\ell-1}(\mathcal C)$. 
\end{itemize}

Our main result Theorem~\ref{main} yields

\begin{Cor}
If $\mathcal C$ is a class of strongly connected automata satisfying the {\v C}ern\'{y} conjecture, then every automaton $\A\in I_{\ell}(\mathcal C)$, for some $\ell \geq 0$, satisfies the {\v C}ern\'{y} conjecture.
\end{Cor}

In particular, if $\mathcal C_0$ consists of one-element automata, then the class $I_{\ell}(\mathcal C_0)$ contains, in particular, all strongly connected weakly monotonic automata of level $\ell$. Yet, one may start from a broader class of automata $\mathcal C$ for which the {\v C}ern\'{y} conjecture has been already verified. Then, one may easily construct many new examples of strongly connected automata not covered by the results on the {\v C}ern\'{y} conjecture established so far.

\end{document}